\documentclass[a4paper,11pt]{article}
\usepackage[margin=26mm]{geometry}
\usepackage{amsthm}
\usepackage{amsmath}
\usepackage{mathbbol}
\usepackage{amsfonts}
\usepackage{amssymb}
\usepackage{graphicx}
\usepackage{url,hyperref}
\usepackage{epsfig}
\usepackage{wrapfig}
\usepackage{color}
\usepackage{fancybox}
\usepackage{xcolor}

\usepackage{subfig}

\usepackage{algo} 

\usepackage{comment}

\usepackage[mathlines]{lineno}

\newcommand*\patchAmsMathEnvironmentForLineno[1]{%
  \expandafter\let\csname old#1\expandafter\endcsname\csname #1\endcsname
  \expandafter\let\csname oldend#1\expandafter\endcsname\csname end#1\endcsname
  \renewenvironment{#1}%
     {\linenomath\csname old#1\endcsname}%
     {\csname oldend#1\endcsname\endlinenomath}}%
\newcommand*\patchBothAmsMathEnvironmentsForLineno[1]{%
  \patchAmsMathEnvironmentForLineno{#1}%
  \patchAmsMathEnvironmentForLineno{#1*}}%
\AtBeginDocument{%
\patchBothAmsMathEnvironmentsForLineno{equation}%
\patchBothAmsMathEnvironmentsForLineno{align}%
\patchBothAmsMathEnvironmentsForLineno{flalign}%
\patchBothAmsMathEnvironmentsForLineno{alignat}%
\patchBothAmsMathEnvironmentsForLineno{gather}%
\patchBothAmsMathEnvironmentsForLineno{multline}%
}

\newcommand{\KK}{\mathbb{K}}
\newcommand{\RR}{\mathbb{R}}
\renewcommand{\SS}{\mathbb{S}}

\newcommand{\KKnew}{\mathbb{K}_{\rm new}}
\newcommand{\calKnew}{\mathcal{K}_{\rm new}}
\newcommand{\Gnew}{G_{\rm new}}
\newcommand{\Tnew}{T_{\rm new}}

\newcommand{\calK}{\mathcal{K}}

\newcommand{\ignore}[1]{}

\newcommand{\freespace}{\mathbb{F}}
\newcommand{\configspace}{\mathbb{X}}

\newcommand{\cycle}{\mathit{cycle}}

\def\find{\mbox{\sc Find}}
\def\findext{\mbox{\sc FindExt}}
\def\union{\mbox{\sc Union}}
\def\unionext{\mbox{\sc UnionExt}}
\def\makeset{\mbox{\sc MakeSet}}
\def\rank{\mathit{rank}}
\def\parent{\mathit{parent}}
\def\parity{\mathit{parity}}
\newcommand\CR{\mbox{\tt cr}_2}		  

\def\DEF#1{\textbf{\emph{#1}}}

\newtheorem{theorem}{Theorem}
\newtheorem{lemma}[theorem]{Lemma}

\title{Semi-dynamic connectivity in the plane}
\author{
	Sergio Cabello\thanks{Department of Mathematics, IMFM, and
			Department of Mathematics, FMF, University of Ljubljana, Slovenia.
			\texttt{sergio.cabello@fmf.uni-lj.si}.
        	Supported by the Slovenian Research Agency, program P1-0297.}
\and
	Michael Kerber\thanks{Max-Planck-Institut f\"ur Informatik, Saarbr\"ucken, Germany.
                              \texttt{mkerber@mpi-inf.mpg.de}.
                              Supported by the Max Planck Center for Visual Computing and Communication.}
}
\date{\today}

\begin{document}
\maketitle

\begin{abstract}
	Motivated by a path planning problem we consider the following procedure.
	Assume that we have two points $s$ and $t$ in the plane and take $\calK=\emptyset$.
	At each step we add to $\calK$ a compact convex set that does not
	contain $s$ nor $t$.
	The procedure terminates when the sets in $\calK$ separate $s$ and $t$.
	We show how to add one set to $\calK$
	in $O(1+k\alpha(n))$ amortized time plus the time needed 
	to find all sets of $\calK$ intersecting the newly added set,
	where $n$ is the cardinality of $\calK$,
	$k$ is the number of sets in $\calK$ intersecting the newly added set, and 
	$\alpha(\cdot)$ is the inverse of the Ackermann function.
\end{abstract}

\section{Introduction}

Consider the \emph{path planning problem} from robotics, 
also known as the \emph{piano mover's problem}~\cite{lavalle}~\cite[Ch.13]{dutchbook}:
Given an initial and a target configuration of a robot, the task
is to decide whether the robot can move from the initial to the target
configuration without colliding with itself or a surrounding object
(and to find such a transformation if it exists).
The problem is typically tackled by setting up a \DEF{configuration space} $\configspace$
where every robot position is encoded as a single point.
Then $\configspace$ is partitioned into a \DEF{free space} $\freespace\subseteq\configspace$ 
of allowed configurations and its complement 
$\bar{\freespace}= \configspace\setminus\freespace$ denoting configurations
that collide with obstacles. The initial and final state are denoted by two points $s$
and $t$ in $\freespace$, and the task is to decide whether $s$ and $t$ are in the same
path-connected component of $\freespace$.

The following approach to solve the path planning problem
is discussed by Wang, Chiang and Yap~\cite{wcy-soft}. 
Assume for simplicity
that the configuration space $\configspace$ is a unit cube in $\RR^d$.
For any given subcube, which we call \DEF{box} from now, we can decide whether the box 
is entirely contained in $\freespace$, entirely contained in $\bar{\freespace}$,
or both contains points of $\freespace$ and $\bar{\freespace}$. We color a box
green, red, or yellow, respectively, depending on the predicates outcome.
Now, starting with the entire $\configspace$, we build a quadtree structure
and keep subdividing yellow boxes
into $2^d$ boxes of equal size until one of the following events occur:

\begin{itemize}
\item[(1)] Points $s$ and $t$ lie in green boxes and are connected by a path that lies entirely in green boxes. Such a path is a solution to the path planning problem.
See Figure~\ref{fig:problem_illustration}, left, for an illustration.
\item[(2)] Each path from $s$ to $t$ intersects some red square.
In this case, no collision-free path from $s$ to $t$ can exist, and
we say that the red boxes separate $s$ and $t$.
See Figure~\ref{fig:problem_illustration}, right, for an illustration.
\end{itemize}

The described subdivision strategy is also used for the task of segmentation of digital images; see~\cite{aizawa} and references therein.
In that situation, the approach would decide whether the pixels $s$ and $t$ belong to the same connected component of the image.

How quickly can we decide whether one of the two conditions is satisfied?
Condition (1) can be easily checked by union-find~\cite{tarjan}: just create a new element
for each new green box and make unions to keep together adjacent green boxes, 
always checking whether the boxes containing $s$ and $t$ fall into the same set. 
That means that the amortized complexity
of checking condition (1) is almost linear in the number of green boxes produced.
For condition (2), the case seems less clear~-- an alternative way of phrasing
the condition is to check whether the union of green and yellow boxes contains
$s$ and $t$ in the same connected component. The union-find approach cannot directly be applied
because yellow regions might turn into red and, therefore, the area covered by the boxes
may shrink.
In this paper, we discuss how to test the second condition in the planar case ($d=2$).

\begin{figure}[thb]
\centering
	\includegraphics[width=6cm,page=1]{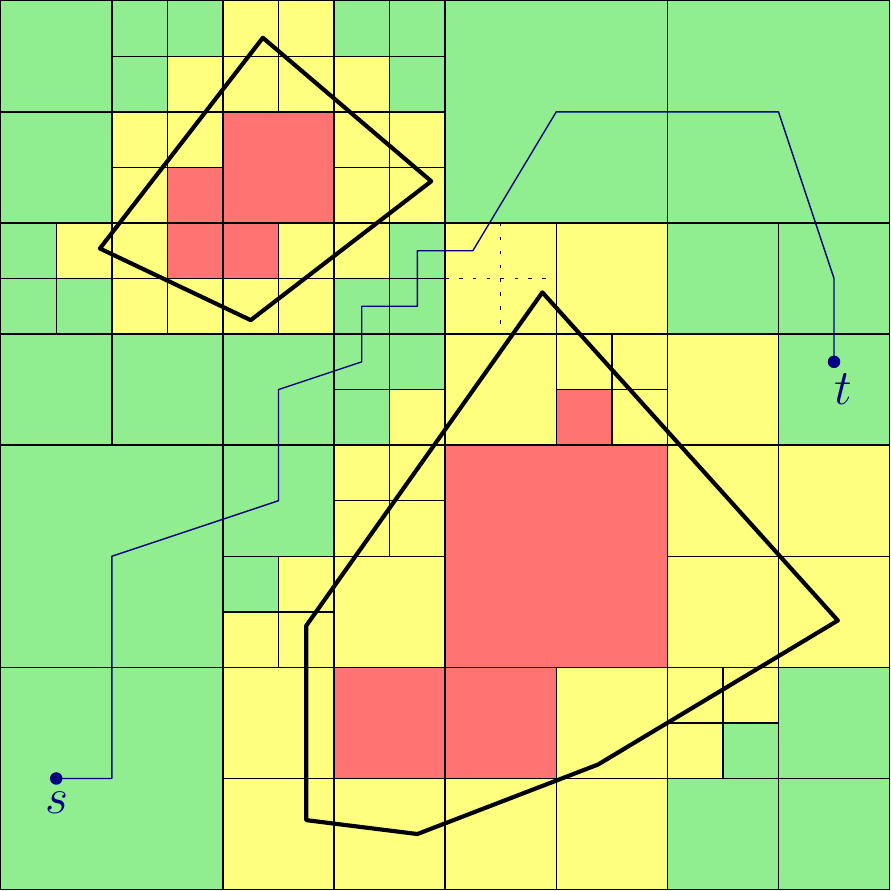}
	\hspace{0.5cm}
	\includegraphics[width=6cm,page=2]{subdivision}
	\caption{Left: Configuration space with two (convex) holes. 
	When subdividing the marked yellow box according to the dashed lines, $s$ and $t$ become connected.
	Right: Configuration space with an annulus-shaped obstacle. When subdividing the marked yellow box, the union
	of red boxes separates $s$ and $t$, so no path can exist.}
	\label{fig:problem_illustration}
\end{figure}

We consider the following generalization of the problem.
We have two points $s$ and $t$ in the plane.
We get a set $\calK$ of compact, convex sets in the plane
iteratively, adding the sets one by one. 
Each of the sets added to $\calK$ is disjoint from $s$ and $t$.
In the motivating problem,
the red boxes would be the elements of $\calK$.
At the end of the insertion of a new compact convex set into $\calK$, 
we want to know whether $\calK$ separates $s$ and $t$. That is,
we want to know whether each path from $s$ to $t$ has to intersect
some element of $\calK$.
Thus, we want a semi-dynamic data structure to store $\calK$ that
allows the insertion of new elements to $\calK$ and decides whether
$\calK$ separates $s$ and $t$.

We show that we can maintain $\calK$ under insertions 
using a slightly more sophisticated union-find approach. 
The time to insert a new set $K_u$ into $\calK$ is the time
we need to find all the $k$ elements of $\calK$ intersecting $K_u$,
plus $O(k)$ union-find operations.
The idea is based on a classical parity argument saying that 
$s$ and $t$ are separated if and only if we
can find a closed curve contained in the union of the elements of $\calK$ 
that is crossed an odd number of times by the line segment $\ell$ from $s$ to $t$. 
We maintain a union-find data structure for the sets of $\calK$ and augment
it by storing additional information about the parity of crossings with the line segment $\ell$.
Using this additional knowledge, we can quickly decide whether adding a new set
to $\calK$ forms a cycle that separates $s$ and $t$, 
and the information can be maintained under unions
and path compressions without asymptotic overhead.

If in the motivating subdivision procedure we always subdivide a largest
yellow box, we obtain $O(1)$ time per yellow box and 
$O(\alpha(n))$ amortized time per red box, where $n$ is the number of red boxes and
$\alpha(\cdot)$ is the inverse of the Ackermann function.
Thus, we obtain the same asymptotic behavior for testing conditions (1) and (2).

\paragraph{Roadmap.}
In Section~\ref{sec:static} we discuss a criterion to decide when $\calK$
separates $s$ and $t$ in the static case.
In Section~\ref{sec:dynamic} we extend this to the semi-dynamic case,
where sets get added to $\calK$.
In Section~\ref{sec:subdivision} we discuss the application to the motivating
subdivision procedure.

Our aim is to provide a self-contained exposition.
Some of the arguments are an adaptation of Cabello and Giannopoulos~\cite{cg-15} 
to this simpler setting,
others can be shorten substantially using machinery from Algebraic Topology.

\section{Static connectivity}
\label{sec:static}

Let $\calK$ denote a finite family of compact convex sets in the plane,
and let $\KK$ denote their union. 
We use the notation $\bar \KK =\RR^2\setminus \KK$.
Let $s$ and $t$ be points in $\bar\KK$.

The set $\KK$ \DEF{separates}
$s$ and $t$ if they are in different path-connected components of $\bar \KK$.
Equivalently, $\KK$ separates $s$ and $t$ if 
each path in the plane from $s$ to $t$ intersects $\KK$.
We also say that $\calK$ separates $s$ and $t$.

In the next subsection we discuss a criterion to decide when $\KK$ separates $s$ and $t$.
The criterion is based on considering all polygonal paths contained in $\KK$, and thus is
computationally unfeasible.
In Subsection~\ref{sec:intersection} we discuss how this criterion can be checked 
in the intersection graph of $\calK$, and thus obtain a discrete version suitable
for computations.

We will consistently use Greek letters $\pi,\gamma, \tau, \dots$ only for (polygonal) curves.

\subsection{Topological criterion for separation} 

A polygonal curve $\pi$ is \DEF{generic} (with respect to $s$ and $t$) if $\pi$ does not contain
$s$ nor $t$ and the line segment from $s$ to $t$ does neither contain an endpoint
of $\pi$ nor a self-intersection of $\pi$. 
We will assume in our discussion that all the polygonal curves are generic.
We can enforce this assumption making a rotation, so that $\ell$ is horizontal,
and replacing the point $s$ by $s'=s+(0,\varepsilon)$, for an infinitesimal $\varepsilon>0$.
We always use the same perturbed point $s'$.
Since $\calK$ is finite,
separation of $s$ and $t$ with $\KK$ is equivalent to separation of $s'$ and $t$ with $\KK$.
The computations can then be made using simulation of simplicity~\cite{sos}.

We fix $\ell$ as the line segment joining $s'$ and $t$.
The \DEF{crossing number} of $\ell$ with a polygonal curve $\pi$
is the number of intersections of $\ell$ and $\pi$.
We denote by $\CR(\ell,\pi)$ the modulo $2$ value of the crossing number of $\ell$ and $\pi$.
Thus, $\CR(\ell,\pi)=1$ if and only if the crossing number is odd.
For the whole paper, \emph{any arithmetic involving $\CR(\cdot,\cdot)$ is done modulo 2.}

It is important to use always the same perturbed point $s'$. Then, if a polygonal curve $\pi$ is 
the concatenation of $\pi'$ and $\pi''$, we have
$\CR(\ell,\pi) = \CR(\ell,\pi') + \CR(\ell,\pi'')$. If we would use different perturbed points
and the common endpoint of $\pi'$ and $\pi''$ lies in the line segment $st$,
then the inequality does not necessarily hold.

A polygonal curve $\pi$ is \DEF{closed} if its endpoints coincide.
It is \DEF{simple} if it does not have any self-intersections, except for the common
endpoint in the case of closed polygonal paths. 

Note that in the following lemma we do not require simple curves.

\begin{lemma}
\label{lem:parity_folklore}
	The set $\KK$ separates $s$ and $t$ if and only if
	there exists a closed polygonal curve $\pi$ contained in $\KK$ 
	such that $\CR(\ell,\pi)=1$.
\end{lemma}
\begin{proof}
	We use the following classical argument, 
	which sometimes is an intermediary step towards a proof of the Jordan's curve theorem:
	A simple closed polygonal curve $\pi$ separates $s'$ and $t$ if and only if 
	$\ell$ and $\pi$ have an odd crossing number. 
	See, for example, Mohar and Thomassen~\cite[Section 2.1]{mt-01} for a formal proof. 

	Assume that $\KK$ contains a closed polygonal curve $\pi$ 
	such that $\ell$ and $\pi$ have an odd crossing number. 
	The curve $\pi$ may have self-intersections.
	If $\pi$ is not simple, we can split it at self-intersections arbitrarily to obtain
	simple, closed polygonal paths $\pi_1,\dots, \pi_k$ that have, all together,
	the same image as $\pi$. 
	Since we have $1=\CR(\ell,\pi)=\sum_i \CR(\ell,\pi_i)$,
	at least one of the curves $\pi_i$ has $1=\CR(\ell,\pi_i)$.
	Such a curve $\pi_i$ separates the endpoints of $\ell$, and thus separates $s$ and $t$.
	It follows that there is no path in $\RR^2\setminus \pi_i$ from $s$ to $t$.
	Since $\bar\KK\subset \RR^2\setminus \pi_i$, there is no path in $\bar\KK$
	from $s$ to $t$.

	Assume that there is no path in $\bar\KK$ from $s$ to $t$. 
	Consider the path-connected component $A$ of $\bar\KK$ that contains $s$. 
	Since $t$ is in a different cell of $\bar\KK$
	and $\calK$ is a finite collection of compact, convex bodies, 
	there exists a simple closed curve $\pi$ contained in the boundary of $A$ that separates $s$ and $t$.
	We can make shortcuts in $\pi$ to obtain a simple closed polygonal curve $\pi'$ contained in $\KK$
	that separates $s$ and $t$. (This can be shown formally using the convexity of
	the elements of $\calK$ and the compactness of $\KK$.)
	The resulting simple polygonal path $\pi'$ separates $s$ and $t$, and thus
	the crossing number of $\ell$ and $\pi'$ is odd.
\end{proof}

\begin{lemma}
\label{lem:same}
	Let $K_u$ and $K_v$ be two compact convex sets of $\calK$.
	For any two generic polygonal curves $\pi$ and $\pi'$ contained in $K_u\cup K_v$ 
	with the same endpoints, we have $\CR(\ell,\pi)=\CR(\ell,\pi')$.
\end{lemma}
\begin{proof}
	First note that $K_u\cup K_v$ does not separate $s$ and $t$. This can be 
	seen as follows. Let $\SS^1$ be the set of directions.
	Consider the set of directions of the vectors $\overrightarrow{sx}$,
	for all $x\in K_u$. Since $K_u$ is convex and $s\notin K_u$, this directions
	cover less than half of $\SS^1$. A similar statement holds for $K_v$.
	It follows that there exists some ray from $s$ to infinity
	in $\RR^2\setminus (K_u\cup K_v)$. Similarly, there exists a ray from $t$ to infinity
	in $\RR^2\setminus (K_u\cup K_v)$. Those two rays and an extra path far enough
	can be combined to obtain a path from $s$ to $t$ in $\RR^2\setminus (K_u\cup K_v)$. Thus,
	$K_u\cup K_v$ does not separate $s$ and $t$. 

	Since $K_u\cup K_v$ does not separate $s$ and $t$, 
	Lemma~\ref{lem:parity_folklore} implies that any closed path $\gamma$ contained
	in $K_u\cup K_v$ has $\CR(\ell,\gamma)=0$.
	The concatenation of $\pi$ and the reverse of $\pi'$ is a closed path contained
	in $K_u\cup K_v$ and therefore $\CR(\ell,\pi)+\CR(\ell,\pi')=0$.
\end{proof}

\subsection{Criterion on the Intersection Graph} 
\label{sec:intersection}

Consider the \DEF{intersection graph} of $\cal K$ and denote it by $G$.
Each element $K_v\in\calK$ is a node of $G$; we will denote the node by $v$ to 
match standard graph theory notation.
There is an edge $uv$ in $G$ if and only if $K_u$ and $K_v$ intersect. 
The graph $G$ is an abstract graph. 
Next we provide a geometric representation.

For each node $v$ of $G$ choose a point $p_v$ in $K_v$. 
For each edge $uv$ of $G$, let $\gamma(uv)$ be a polygonal path from $p_u$ to $p_v$
contained in the union $K_u\cup K_v$. Since $K_u$ and $K_v$ are convex and intersect, 
we can always choose $\gamma(uv)$ with at most $2$ edges.
The pair 
\[
	(\{ p_v\mid v\in V(G)\},\{ \gamma(uv)\mid uv\in E(G)\})
\] 
is a drawing of $G$.
(It is not necessarily an embedding because drawings of edges may cross, for example when four 
axis-parallel squares have disjoint interiors but share a vertex.)
For each walk $W=e_1\dots e_k$ in $G$, let $\gamma(W)$ be the polygonal path obtained
by concatenating $\gamma(e_1),\dots, \gamma(e_k)$. 
If $W$ is a closed walk, then $\gamma(W)$ is a closed polygonal curve.

\begin{lemma}
\label{lem:nerve}
	The set $\KK$ separates $s$ and $t$ if and only if
	there exists a closed walk $W$ in $G$ such that $\CR(\ell,\gamma(W))=1$.
\end{lemma}
\begin{proof}
	Assume that $\KK$ separates $s$ and $t$.
	Because of Lemma~\ref{lem:parity_folklore}, 
	there is some polygonal curve $\pi$ contained in $\KK$ such that $\CR(\ell,\pi)=1$. 
	We break the path $\pi$ into pieces such that each piece is contained
	in the union of $2$ sets from $\cal K$.
	Let $\pi_1,\dots ,\pi_k$ be the resulting pieces, each of them a polygonal curve. 
	For each piece $\pi_i$, let $x_i$ and $y_i$ be the endpoints
	of $\pi_i$, and let $K_{u_i}$ and $K_{v_i}$ be the elements of $\cal K$ that contain $x_i$ and $y_i$,
	respectively, so that $\pi_i$ is contained in $K_{u_i}\cup K_{v_i}$.
	Note that $u_i v_i$ is an edge of $G$. 
	Let $W$ be the closed walk with edges $u_1v_1,\dots, u_kv_k$.
	
	We claim that $\CR(\ell,\gamma(W))=\CR(\ell,\pi)=1$. 
	To see this, consider for each piece $\pi_i$ the polygonal curve 
	$\hat\gamma_i$ from $p_{u_i}$ to $p_{v_i}$ obtained by concatenating 
	the line segment from $p_{u_i}$ to $x_i$, 
	followed by $\pi_i$, 
	and followed by the line segment from $y_i$ to $p_{v_i}$. 
	See Figure~\ref{fig:notation} for an example.
	For each piece $\pi_i$, the polygonal curves $\hat\gamma_i$ and $\gamma(u_iv_i)$ have the same
	endpoints and are contained in the union $K_{u_i}\cup K_{v_i}$.
	Because of Lemma~\ref{lem:same}, we have
	$\CR(\ell,\hat\gamma_i)=\CR(\ell,\gamma(u_iv_i))$.
	It follows that, if we define $\hat\gamma$ as the concatenation of $\hat\gamma_1,\dots ,\hat\gamma_k$,
	we have $\CR(\ell,\gamma(W))=\CR(\ell, \hat\gamma)$. 
	Moreover, $\CR(\ell,\hat\gamma)=\CR(\ell, \pi)$
	because $\hat\gamma$ is essentially $\pi$ with some spokes connecting $x_i$ to $p_{u_i}$,
	where the number of crossings evens out.
	We conclude that $\CR(\ell,\gamma(W))= \CR(\ell,\hat\gamma)=\CR(\ell, \pi)=1$.
	This finishes one direction of the proof.
	
	\begin{figure}[thb]
	\centering
		\includegraphics[page=1,scale=.9]{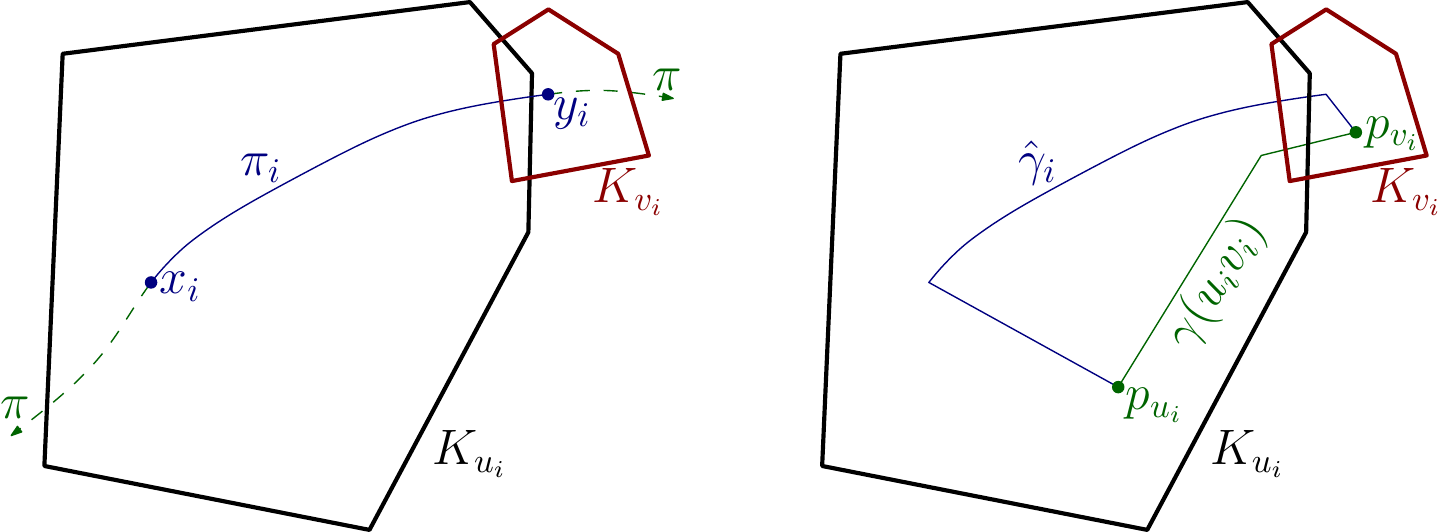}
		\caption{Notation in the proof of Lemma~\ref{lem:nerve}.}
		\label{fig:notation}
	\end{figure}

	For the other direction, assume that $G$ has a closed walk $W$ such that the crossing
	number of $\ell$ and $\gamma(W)$ is odd. Since the closed polygonal path
	$\gamma(W)$ is contained in $\KK$ by
	construction, Lemma~\ref{lem:parity_folklore} implies that $\KK$ separates $s$ and $t$.
	This proves the other direction.
\end{proof}

We extend Lemma~\ref{lem:nerve} 
to a necessary and sufficient condition for $s$ and $t$ being disconnected
that involves only a few cycles of $G$. 
Let $T$ be any maximal spanning forest of $G$, that is, $T$ contains a spanning
tree of each connected component of $G$. 
For each edge $e$ of $G-E(T)$, let $\cycle(T,e)$ be the unique cycle
in $T+e$, and let $\tau(T,e)$ be the curve $\gamma(\cycle(T,e))$.
That is, $\tau(T,e)$ is the polygonal curve describing $\cycle(T,e)$ in the drawing.

\begin{lemma}
\label{lem:homology_condition}
	Let $T$ be a maximal spanning forest of $G$.
	The set $\KK$ separates $s$ and $t$ if and only if
	there exists some edge $e\in E(G)\setminus E(T)$ such that $\CR(\ell,\tau(T,e))=1$.
\end{lemma}
\begin{proof}
	The essential idea is to use the so-called cycle space of a graph and the fact that
	$\{ \cycle(T,e) \mid e\in E(G)\setminus E(T)\}$ is a basis. We next provide the details
	using no background.
	
	Since we can treat each connected component of $G$ (and thus $\KK$) independently,
	we will just assume that $G$ has one connected component. This means that $T$ is a spanning
	tree of $G$.
	
	Fix any node $r\in V(G)$ and take the point $p_r\in K_r$ as a basepoint. 
	For each node $v\in V(T)$, let $T[r,v]$ be the simple walk in $T$ from $r$ to $v$.
	For each edge $uv$ of $G$ we define a closed polygonal curve $\lambda(uv)$ 
	as the concatenation of $\gamma(T[r,u])$, $\gamma(uv)$, and the reverse of $\gamma(T[r,v])$.
	Note that $\lambda(uv)$ is a closed polygonal path through $p_r$.
	
	When $uv\notin E(T)$, $\lambda(uv)$ is $\tau(T,uv)$ with a spoke following $\gamma(T[r,w])$,
	where $w$ is the last common node of $T[r,u]$ and $T[r,v]$.
	This implies that 
	\begin{equation}
		\forall uv \in E(G)\setminus E(T): ~~~
			\CR(\ell,\tau(T,uv)) = \CR(\ell,\lambda(uv)).
		\label{eq:noT}
	\end{equation}
	When $uv\in E(T)$, then $\lambda(uv)$ walks twice the same polygonal curve,
	and therefore 
	\begin{equation}
		\forall uv \in E(T): ~~~
			\CR(\ell,\lambda(uv))=0.
		\label{eq:T}
	\end{equation}
	
	Assume that the points $s$ and $t$ lie in different path-components of $\bar\KK$.
	Because of Lemma~\ref{lem:nerve}, there exists some closed walk $W$ in $G$ 
	with $\CR(\ell,\gamma(W))=1$.
	Let $u_1v_1,\dots,u_kv_k$ be the sequence of edges in $W$, where $u_1=v_k$. 
	Using arithmetic modulo $2$ we have  
	\begin{align*}
		1~&=~ \CR (\ell, \gamma(W)) \\
		&=~ \sum_{i=1}^k \CR(\ell,\gamma(u_iv_i))\\
		&=~ \sum_{i=1}^k \Bigl( \CR(\ell,\gamma(T[r,u_i])+ \CR(\ell,\gamma(u_iv_i) + \CR(\ell,\gamma(T[r,v_i]) \Bigr)\\
		&=~ \sum_{i=1}^k \CR(\ell,\lambda(u_iv_i)),
	\end{align*}
	where in the third equality we have used that each node of $W$ 
	is the endpoint of two consecutive edges of $W$, which implies that the new terms
	cancel out.
	This means that, for some edge $u_iv_i$ of $W$, we have
	$\CR(\ell,\lambda(u_iv_i))=1$. This edge $u_iv_i$ cannot be in $T$
	because of \eqref{eq:T}.
	Therefore we have some edge $u_iv_i$ in $E(W)$, where $u_iv_i\notin E(T)$,
	with $\CR(\ell,\lambda(u_iv_i))=1$. Because of \eqref{eq:noT} we have
	$\CR(\ell,\tau(T,u_iv_i))= \CR(\ell,\lambda(u_iv_i))=1$.
	This finishes the proof of one direction of the statement.
	
	For the other direction, assume that there exists some edge $e\in E(G)\setminus E(T)$
	such that $\CR(\ell,\tau(T,e))=1$. Taking $W=\cycle(T,e)$ and using
	that $\tau(T,e)=\gamma(W)$ by definition,
	this means that $W$ is a closed walk in $G$ with $\CR(\ell,\gamma(W))=1$.
	It follows from Lemma~\ref{lem:nerve} that $\KK$ separates $s$ and $t$.	
\end{proof}

\section{Semi-dynamic connectivity}
\label{sec:dynamic}

In this section we discuss the separation of $s$ and $t$ 
under the addition of new sets to $\calK$.
We first describe a standard union-find data structure because we will build on it.
Then we describe the setting and the notation we will use.
It follows a description of the extension of the union-find data structure
for our setting. Finally, we describe the data structure, its maintenance,
and its correctness.

\subsection{Preliminaries: Union-find}

Here we review a standard union-find data structure and some of its properties.
See~\cite[Chapter 21]{cormen}, \cite[Chapter 5]{dpv} or~\cite{e-14} for a comprehensive exposition.

A \DEF{union-find data structure} represents a disjoint set system supporting
the operations \makeset\ (create a new disjoint set with a single element), \union\
(merge two sets), and \find\ (return a representative of a given set). 
We can test whether two elements belong two the same set by testing whether the output
of \find\ for those two elements is the same.
A common realization
is to represent each disjoint set by a rooted tree in which each node holds one element of the set.
The root of the tree holds the representative of the set. Each node has a pointer to its parent,
while the root points to itself.
Then \find\ simply follows the parent pointer until it finds the root of the tree.
The union operation merges two trees by making the root of one subtree a child of the root of the other.
Thus, given two elements, we first locate the roots of their corresponding trees calling \find,
and then we proceed with the union.

Two optimizations are commonly used to obtain an efficient realization. 
\emph{Union-by-rank} determines which root gets merged in a union operation: 
each root has a rank associated to it, in an union we simply make the root of lower rank a child
of the root with larger rank, and we increase the rank of the root if both roots had the same rank. 
\emph{Path compression} makes all nodes found on a search path 
from a node to its root direct children of the root. 
For later reference and modification, 
we include pseudocode in Figure~\ref{fig:code1}.
Combining these two optimizations, each operation has an amortized time
complexity of $\alpha(n)$, where $n$ is the number of elements in the set system and
$\alpha(\cdot)$ is the extremely slow growing inverse Ackermann function.
See references~\cite[Chapter 21]{cormen}, \cite{e-14} or~\cite{ss-05} 
for an analysis of the time complexity.

\begin{figure}[htb]
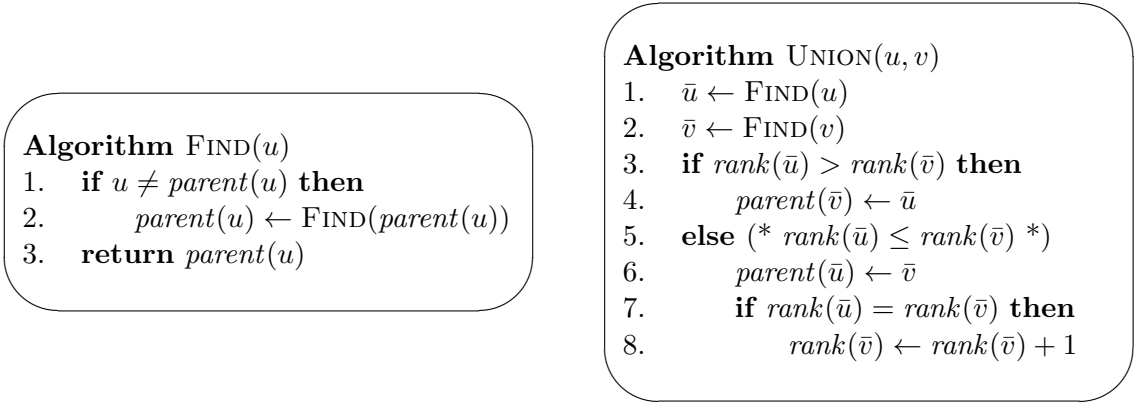

	\hfill
	\ovalbox{
	\parbox{6.6cm}
	{\begin{algorithm}{\find}{$(u)$}
		\qif $u \not= \parent(u)$ \qthen\\
			$\parent(u) \qlet \find (\parent(u))$
		\qfi\\
		\qreturn $\parent(u)$
	\end{algorithm}}}
	\hfill
	\ovalbox{
	\parbox{6.6cm}
	{\begin{algorithm}{\union}{$(u,v)$}
		$\bar u\qlet \find(u)$\\
		$\bar v\qlet \find(v)$\\
		\qif $\rank(\bar u)> \rank (\bar v)$ \qthen\\
			$\parent(\bar v) \qlet \bar u$
		\qelse (* $\rank(\bar u)\le \rank (\bar v)$ *)\\
			$\parent(\bar u) \qlet \bar v$\\
			\qif $\rank(\bar u) = \rank (\bar v)$ \qthen\\
				$\rank(\bar v) \qlet \rank (\bar v)+1$
			\qfi
		\qfi
	\end{algorithm}}}
	\hfill
	\caption{The main two operations in the union-find data structure. 
			$u$ and $v$ are nodes of the tree.}
	\label{fig:code1}
\end{figure}

\subsection{Setting}
\label{sec:setting}
Let $s$ and $t$ be two points in the plane.
We have a finite family of convex sets $\calK$, all of them disjoint from $s$ and $t$.
Following the previous notation, we denote by $\KK$ the union of the sets in $\calK$, 
and by $G$ the intersection graph of $\calK$.

Consider the addition of a new compact convex set $K_u$ to $\calK$. 
We use $\calKnew$ for the resulting set, $\KKnew$ for the union of its sets,
and $\Gnew$ for the intersection graph of $\calKnew$.
 
The analysis of our data structure is based on a maximal spanning forest 
of the intersection graph of the convex sets.
The definition of such spanning forest is iterative, as follows.
Let $uv_1,\dots, uv_k$ be an enumeration of the edges incident to $u$ in $\Gnew$.
That is, $K_{v_1},\dots,K_{v_k}$ are the sets of $\calK$ intersecting the new set $K_u$.
We consider adding the edges $uv_1,\dots, uv_k$ to $G$ one by one.
We thus define $G_0$ as the union of $G$ and a new vertex $u$ for $K_u$.
For each index $1\le j\le k$, we define the graph $G_j=G_{j-1}+ uv_j$.
Note that $\Gnew=G_k$. The intermediate graphs $G_1,\dots,G_{k-1}$ are not
intersection graphs of $\calK$ or $\calKnew$, but something in between.

If at the time of adding $uv_j$ the vertices $u$ and $v_j$ are already connected 
in the graph $G_{j-1}$, then we call $uv_j$ a \DEF{cycle edge}.
Otherwise, $uv_j$ merges two connected components of $G_{j-1}$
and we call it a \DEF{merge edge}.
Note that whether an edge is a cycle edge or a merge edge depends on the order used in the addition of edges.

Let $T$ be the maximal spanning forest of $G$.
We define $T_0$ as the union of $T$ and a new vertex $u$ for $K_u$. 
For each index $1\le j\le k$ we define
\[
	T_j~=~ \begin{cases}
			T_{j-1} &\text{if $uv_j$ is a cycle edge,}\\
			T_{j-1}+uv_j & \text{if $uv_j$ is a merge edge.}
		\end{cases}
\]
It is easy to see by induction that, for each index $1\le j\le k$,
$T_j$ is a maximal spanning forest of $G_j$.
We define $\Tnew$ as $T_k$. Thus $\Tnew$ is a maximal spanning
forest of $\Gnew=G_k$.

As it was done in Section~\ref{sec:intersection},
for each $K_u$ we choose a point $p_u$ in $K_u$ and for
each edge $uv$ we choose a polygonal curve $\gamma(uv)$.
These choices are made in the first appearance of the node or edge,
and remain invariant from there onwards.

\subsection{Augmented union-find}
\label{sec:extended}

We maintain a union-find data structure for the connected components 
of the graphs $G_j$.
Recall that $T_j$ is a maximal spanning forest of $G_j$.
For each node $v$ of $G_j$, we store a \DEF{parity bit}, denoted as $\parity(v)$, 
with the following property:
\begin{itemize}
	\item If $v$ is the root of a union-find tree, then $\parity(v)=0$.
	\item If $v$ has parent $w$ in a union-find tree, 
		then $\parity(v)= \CR(\ell, T_j[w,v])$. That is,
		we look at the parity of the crossing number of $\ell$ with
		the polygonal curve from $p_v$ to $p_w$ defined by the drawing of $T_j$.
\end{itemize}
For the rest of the paper, \emph{any arithmetic involving parity bits is done modulo 2.}

We next argue that the correct parity bits can be maintained in the same complexity as the union-find operations, assuming that only certain unions are made.
That is clear for \makeset\ by giving the new node parity $0$. 

Consider the \find\ operation, which changes parent pointers due to path compression. 
Note that the graphs $G_j$ and $T_j$ do not change, but the union-find data structure does.
Let $u,v,w$ be nodes such that, in the union-find data structure,
$w$ is parent of $v$ and $v$ is parent of $u$.
Note that 
\[
		\CR(\ell, \gamma(T_j[u,w])) ~=~
		\CR(\ell, \gamma(T_j[u,v])) + \CR(\ell, \gamma(T_j[v,w])) ~=~
		\parity(u) + \parity(v).
\]
Therefore, when we update $\parent(u)\qlet w$,
we just have to set $\parity(u) \qlet \parity(u) + \parity(v)$
to restore $\parity(u)$ to its correct value.

We can now easily realize the augmented path compression.
We define an extended function $\findext(u)$ that, for 
all nodes $v$ from $u$ to the root $r$ of the tree containing $u$,
sets $\parent(u)=r$ and updates the value $\parity(v)$ accordingly. 
Pseudocode is given in Figure~\ref{fig:code2}.
It easily follows by induction that \findext\ correctly maintains 
the parity bit of all elements.

\begin{figure}[htb]
	\centering
	\ovalbox{
	\parbox{9cm}
	{\vspace{.4cm}\begin{algorithm}{\findext}{$(u)$}
		\qif $u \not= \parent(u)$ \qthen\\
			$r \qlet \find (\parent(u))$\\
			$\parity(u) \qlet \parity(u) + \parity(\parent(u))$\\
			$\parent(u)\qlet r$
		\qfi\\
		\qreturn $\parent(u)$
	\end{algorithm}}}
	\caption{Extended find operation for an element $u$.}
	\label{fig:code2}
\end{figure}

Finally, we discuss the extension \union\ to \unionext. 
Its arguments are two nodes $u$ and $v_j$ such that
$uv_j$ is a merge edge and the union-find data structure stores the
connectivity of $G_{j-1}$. Since $uv_j$ is a merge edge,
we have $T_j=T_{j-1}+uv_j$.
This means that the sets $K_u$ and $K_{v_j}$ intersect but 
$u$ and $v_j$ were in different connected components of $G_{j-1}$.
Like before, we first find the roots $\bar u$ and $\bar v$ of their trees 
using $\findext(\cdot)$. After this it holds that
$\parity(u)= \CR(\ell, \gamma(T_{j-1}[\bar u,u]))$, and similarly
$\parity(v)= \CR(\ell, \gamma(T_{j-1}[\bar v,v]))$.

The walk $T_{j}[\bar u, \bar v]$
can be split into $T_{j-1}[\bar u, u]$, $uv$, and $T_{j-1}[v, \bar v]$.
Thus we have
\begin{align*}
		\CR(\ell, \gamma(T_j[\bar u,\bar v])) ~&=~ 
			\CR(\ell, \gamma(T_{j-1}[u,\bar u])) + \CR(\ell, \gamma(uv)) 
			+ \CR(\ell, \gamma(T_{j-1}[v,\bar v])) \\
			&=~ \parity(u) + \CR(\ell, \gamma(uv)) + \parity (v).
\end{align*}
The last values are either available through $\parity(\cdot)$ 
or computable in constant time. 
If, for example, $\bar u$ gets $\bar v$ as its parent, then
we have $\parity(\bar u)=\CR(\ell, \gamma(T_{j}[u,v]))$. The other case
is similar. We provide the resulting pseudocode in Figure~\ref{fig:code3}.

\begin{figure}[htb]
	\centering
	\ovalbox{
	\parbox{8cm}
	{\vspace{.4cm}\begin{algorithm}{\unionext}{$(u,v)$}
		$\bar u\qlet \findext(u)$\\
		$\bar v\qlet \findext(v)$\\
		$b\qlet \parity(u)+\parity(v)+ \CR(\ell, \gamma(uv))$\\
		\qif $\rank(\bar u)> \rank (\bar v)$ \qthen\\
			$\parent(\bar v) \qlet \bar u$\\
			$\parity(\bar v) \qlet b$
		\qelse (* $\rank(\bar u)\le \rank (\bar v)$ *)\\
			$\parent(\bar u) \qlet \bar v$\\
			$\parity(\bar u) \qlet b$\\
			\qif $\rank(\bar u) = \rank (\bar v)$ \qthen\\
				$\rank(\bar v) = \rank (\bar v)+1$
			\qfi
		\qfi
	\end{algorithm}}}
	\caption{Extended union for two nodes $u$ and $v$.}
	\label{fig:code3}
\end{figure}

The properties of union-find imply that each of the extended operations,
\unionext\ and \findext, has an amortized
complexity of $\alpha(n)$, where $n$ is the cardinality of $\calK$.

\subsection{Semi-dynamic data structure}

We now describe the data structure to maintain $\calK$.
The data structure supports one operation: add a new
compact convex set $K_u$ to $\calK$ and then 
report whether $\calK\cup \{ K_u\}$ separates $s$ and $t$.
We use the notation from Sections~\ref{sec:setting} and~\ref{sec:extended}.

The data structure has the following elements:
\begin{itemize}
	\item an augmented union-find data structure as described in Section~\ref{sec:extended};
	\item for each element $K_v$ of $\calK$, we store the point $p_v$;
	\item a semi-dynamic data structure $DS(\calK)$ that can find, for the new $K_u$, all the 
		objects of $\calK$ that intersect $K_u$.
\end{itemize}
The intersection graph $G$ and the maximal spanning forest $T$ are not kept.
They are used only for the analysis. 

We next describe how to insert $K_u$.
We use the data structure $DS(\calK)$ to find the sets $K_{v_1},\dots,K_{v_k}$ of $\calK$
that intersect $K_u$. We then insert $K_u$ in the data structure $DS(\calK)$ to obtain 
$DS(\calKnew)$. We choose a point $p_u$ in $K_u$
and create a new node $u$ in the extended union-find data structure.

We then iterate over the edges $uv_1,\dots, uv_k$. 
We first decide whether the considered edge $uv_j$ is a merge edge or a cycle edge 
by checking whether 
$\findext(u)$ and 
$\findext(v_j)$ return the same representative.
If $uv_j$ is a merge edge, we just call $\unionext(u,v_j)$ and continue with the next step
of the filtration.

Otherwise, $uv_j$ is a cycle edge, and we proceed as follows. 
We want to check whether $\CR(\ell,\tau(T_{j},uv_j))=\CR(\ell,\tau(T_{j-1},uv_j))$ is $1$ or $0$.
For this, we use that $u$ and $v$ have already the same parent because of the calls 
$\findext(u)$ and $\findext(v)$.
If we denote such a common parent by $r$, then
\begin{align*}
	\CR(\ell,\tau(T_j,uv_j)) ~&=~ \CR(\ell,\gamma(T_{j-1}[u,v_j])) + \CR(\ell,\gamma(uv_j)) \\
			&=~ \CR(\ell,\gamma(T_{j-1}[u,r])) + \CR(\ell,\gamma(T_{j-1}[v_j,r])) + \CR(\ell,\gamma(uv_j)) \\
			&=~ \parity(u) + \parity(v_j) + \CR(\ell,\gamma(uv_j)).
\end{align*}
If $\CR(\ell,\tau(T_j,uv_j))=1$, then we conclude 
that $\KKnew$ separates $s$ and $t$ and we finish the algorithm.
If $\CR(\ell,\tau(T_j,uv_j))=0$, we proceed to the next edge $uv_{j+1}$.
Pseudocode for the insertion of $K_u$ is given in Figure~\ref{fig:code4}.
This finishes the description of the algorithm.

\begin{figure}[htb]
	\centering
	\ovalbox{
	\parbox{10.5cm}
	{\vspace{.4cm}\begin{algorithm}{Adding $K_u$ to $\cal K$}{}
		$\makeset(u)$\\
		$\parity(u)\qlet 0$\\
		\qfor $K_v\in \calK$ intersecting $K_u$ \qdo\\
			$\bar u\qlet \findext(u)$\\
			$\bar v\qlet \findext(v)$\\
			\qif $\bar u \not= \bar v$ \qthen\\
				$\unionext(u, v)$
			\qelse (* $uv$ a cycle edge *)\\
				$b\qlet \parity(u)+\parity(v)+ \CR(\ell, \gamma(uv))$\\
				\qif $b=1$ \qthen\\
					\qreturn ``$\calK\cup \{ K_u\}$ separates $s$ and $t$!!"
				\qfi
			\qfi
		\qrof
	\end{algorithm}}}
	\caption{Procedure for the addition of $K_u$.}
	\label{fig:code4}
\end{figure}

It follows from the invariants of the extended union-find discussed in
Section~\ref{sec:extended}, that we are correctly computing the value
$\CR(\ell,\tau(T_j,uv_j))$. 
If $\CR(\ell,\tau(T_j,uv_j))=1$, then Lemma~\ref{lem:homology_condition}
implies that $\KKnew$ separates $s$ and $t$. From that point on,
we only need to remember that $s$ and $t$ are separated.

If $\CR(\ell,\tau(T_j,uv_j))=0$, then $\CR(\ell,\tau(T,uv_j))$ will
remain $0$ for all future maximal spanning forests $T$.
This is so because the maximal spanning forest we maintain is monotone increasing:
we only add vertices and edges, but never remove anything.
Thus, we never need to check $\CR(\ell,\tau(T,uv_j))$ again later.
In particular, if $\calK$ did not separate $s$ and $t$ and 
we have $\CR(\ell,\tau(T_j,uv_j))=0$ for all $j$,
then
\[
	\forall vv'\in E(\Gnew)\setminus E(\Tnew):~~~ \CR(\ell,\tau(\Tnew,vv'))=0.
\]
Since $\Tnew$ is a maximal spanning forest of $\Gnew$,
Lemma~\ref{lem:homology_condition} implies that $\KKnew$ 
does not separate $s$ and $t$. 

For each edge $uv_j$ we make 2 calls to \findext, at most one call to \unionext,
and additional $O(1)$ work. This means that for each edge we spend
$O(\alpha(n))$ amortized time, where $n$ is the cardinality of $\calK$.
We also need the time needed to find the elements of $\calK$ intersecting the new
element $K_u$. We conclude.

\begin{theorem}
\label{thm:main}
	Let $s$ and $t$ be two points in the plane.
	There is a semi-dynamic data structure to maintain a family $\calK$ of $n$ compact convex
	sets in the plane under insertions to decide whether $\calK$ separates $s$ can $t$.
	The insertion of a new set $K_u$ in $\calK$ that intersects $k$ sets of $\calK$ 
	takes $O(1+k \alpha(n))$ amortized time, plus the time needed to find
	the $k$ elements of $\calK$ intersecting $K_u$.
\end{theorem}

Of course, once $s$ and $t$ are separated by $\calK$, the insertion of each new
set can be carried out in constant time, since we only need to remember that
$\calK$ separates $s$ and $t$.

\section{Application to dynamic connectivity under subdivision}
\label{sec:subdivision}

We consider now the motivating application discussed in the Introduction for $d=2$.

We have two points $s$ and $t$ inside the unit square $\configspace$.
Initially, the box $\configspace$ is colored yellow.
In each iteration, we take a \emph{largest} yellow box, subdivide it into
4 subboxes, and color each of them as red, yellow, or green 
depending on the outcome of some oracle.
The boxes containing $s$ or $t$ are always colored yellow or green.
We want to know at which point the red boxes separate $s$ and $t$, meaning
that each path from $s$ to $t$ contained in the unit square 
intersects some red box. 

Boxes are assumed to contain their boundary, 
so that any two boxes intersect if their boundaries intersect, 
possibly only at a common vertex.

For our arguments it is convenient to surround $\configspace$ with 8 red boxes of the same size
as $\configspace$. This reduces the problem to finding certain
curves within the red region. Without those additional squares, we should also consider 
boundary-to-boundary curves.

We maintain through the algorithm the intersection graph $H$ of the yellow and red boxes.
This intersection graph $H$ has one node for each box that is yellow or red,
and an edge between two nodes whenever the corresponding boxes intersect.
The graph $H$ is stored using an adjacency list representation~\cite[Chapter 22]{cormen}.
The adjacency list of each vertex is stored as a doubly linked list.
Moreover, for the appearance of a node $v$ in the adjacency list of $u$,
we keep a pointer to the appearance of $u$ in the adjacency list of $v$.
With this, we can perform the deletion of a node $v$ in time proportional to its
degree.

When we want to subdivide a yellow box $K_u$ represented by a node $u$, 
we can locate its set of neighbors $N=N_H(u)$ in the graph $H$, delete $u$ from
the graph, subdivide $K_u$ into four boxes, create the at most four new nodes
representing the yellow and red boxes arising from the subdivision of $K_u$,
check for intersection each of them against each of the nodes in $N$, 
and update the graph $H$ accordingly. 
All this takes time $O(1+|N|)$ time.

If we always subdivide a largest yellow box, there are at most $12$ 
other boxes intersecting it. This means that we can update the 
intersection graph $H$ of yellow and red boxes in $O(1)$ time.
For choosing always a largest yellow box, we can use for example a queue
for the yellow boxes.
Thus, we spend $O(1)$ time per subdivided yellow box and,
for each red box, we get its neighboring red boxes in $O(1)$ time.
Using Theorem~\ref{thm:main} for the red boxes, and a normal union-find
for the green boxes, as discussed in the Introduction,
we obtain the following result.

\begin{theorem}
	Consider the subdivision procedure described in the Introduction where
	we always subdivide a largest yellow box.
	We can perform the subdivision until condition (1) or (2) occurs in $O(n\alpha (n))$ time,
	where $n$ is the number of subdivisions performed.	
\end{theorem}

Of course we can also perform the first $n$ steps of the subdivision procedure
in $O(n\alpha (n))$ time, and correctly report that neither condition (1) nor (2) hold.

\paragraph{Acknowledgments} We thank Chee Yap for posing to us the problem about connectivity under subdivisions.


\begin{thebibliography}{10}

\bibitem{aizawa}
K.~Aizawa, S.~Tanaka, K.~Motomura, and R.~Kadowaki.
\newblock Algorithms for connected component labeling based on quadtrees.
\newblock {\em International Journal of Imaging Systems and Technology}
  19(2):158--166, June 2009, \href{http://dx.doi.org/10.1002/ima.v19:2}%
{doi:10.1002/ima.v19:2}, \url{http://dx.doi.org/10.1002/ima.v19:2}.

\bibitem{dutchbook}
M.~de~Berg, M.~van Kreveld, M.~Overmars, and O.~Schwarzkopf.
\newblock {\em Computational Geometry: Algorithms and Applications}.
\newblock Springer, 2nd edition, 2000.

\bibitem{cg-15}
S.~Cabello and P.~Giannopoulos.
\newblock The complexity of separating points in the plane.
\newblock {\em Algorithmica}, to appear,
  \href{http://dx.doi.org/10.1007/s00453-014-9965-6}%
{doi:10.1007/s00453-014-9965-6}.

\bibitem{cormen}
T.~H. Cormen, C.~E. Leiverson, R.~L. Rivest, and C.~Stein.
\newblock {\em Introduction to Algorithms}.
\newblock MIT Press, 3rd edition, 2009.

\bibitem{dpv}
S.~Dasgupta, C.~H. Papadimitriou, and U.~V. Vazirani.
\newblock {\em Algorithms}.
\newblock McGraw-Hill, 2008.

\bibitem{sos}
H.~Edelsbrunner and E.~P. M{\"{u}}cke.
\newblock Simulation of simplicity: a technique to cope with degenerate cases
  in geometric algorithms.
\newblock {\em {ACM} Transactions on Graphics} 9(1):66--104, 1990,
  \href{http://dx.doi.org/10.1145/77635.77639}%
{doi:10.1145/77635.77639}.

\bibitem{e-14}
J.~Erickson.
\newblock Algorithms notes: Maintaining disjoint sets (``union-find"), 2015.
\newblock Lecture nodes available at
  \url{http://web.engr.illinois.edu/~jeffe/teaching/algorithms/}.

\bibitem{lavalle}
S.~M. Lavalle.
\newblock {\em Planning Algorithms}.
\newblock Cambridge University Press, 2006.

\bibitem{mt-01}
B.~Mohar and C.~Thomassen.
\newblock {\em Graphs on Surfaces}.
\newblock Johns Hopkins University Press, 2001.

\bibitem{ss-05}
R.~Seidel and M.~Sharir.
\newblock Top-down analysis of path compression.
\newblock {\em {SIAM} Journal of Computing} 34(3):515--525, 2005,
  \href{http://dx.doi.org/10.1137/S0097539703439088}%
{doi:10.1137/S0097539703439088}.

\bibitem{tarjan}
R.~E. Tarjan.
\newblock Efficiency of a good but not linear set union algorithm.
\newblock {\em Journal of the ACM} 22(2):215--225, 1975,
  \href{http://dx.doi.org/10.1145/321879.321884}%
{doi:10.1145/321879.321884}.

\bibitem{wcy-soft}
C.~Wang, Y.-J. Chiang, and C.~Yap.
\newblock On soft predicates in subdivision motion planning.
\newblock {\em Proceedings of the Twenty-ninth Annual Symposium on
  Computational Geometry}, pp.~349--358. ACM, SoCG '13, 2013,
  \href{http://dx.doi.org/10.1145/2462356.2462386}%
{doi:10.1145/2462356.2462386}.

\end{thebibliography}
\end{document}